\documentclass[sigplan,10pt]{acmart}
\usepackage{booktabs}   
\usepackage{subcaption} 
\usepackage{cancel}
\usepackage{chngcntr}
\usepackage{thmtools, thm-restate}
\usepackage{amssymb,amstext, amsmath, amsthm}
\usepackage{thm-restate}
\usepackage{hyperref}
\usepackage{soul}
\usepackage{mathtools} 
\usepackage{bcprules} 
\usepackage{mdframed} 
\usepackage{lipsum}
\usepackage{multicol}
\usepackage{graphicx}
\usepackage{xspace}
\usepackage{xfrac} 
\usepackage{caption}
\usepackage{enumerate}
\usepackage{lstautogobble}
\usepackage{listings} 
\usepackage{commands} 
\usepackage[textwidth=1.3cm]{todonotes}

\usepackage{tikz}
\usepackage{pgfplots}

\renewcommand\footnotetextcopyrightpermission[1]{} 
\makeatletter\if@ACM@journal\makeatother
\acmJournal{PACMPL}
\acmVolume{1}
\acmNumber{1}
\acmArticle{1}
\acmYear{2017}
\acmMonth{1}
\startPage{1}
\else\makeatother

\setcopyright{none}             

\newtheorem{defn}{Definition}[section]
\newtheorem{rmk}{Remark}[section]
\newtheorem{conj}{Conjecture}[section]

\begin{document}
\title{Towards Algorithmic Typing for DOT}         
\subtitle{Technical Report}


\author{Abel Nieto}
\affiliation{
  \institution{University of Waterloo}            
  \city{Waterloo}
  \state{ON}
  \country{Canada}
}
\email{anietoro@uwaterloo.ca}          


\thanks{with paper note}                

\begin{abstract}
The Dependent Object Types (DOT) calculus formalizes key features of Scala. The D$_{<:}$ calculus is the core of DOT. To date, presentations of D$_{<:}$ have used declarative typing and subtyping rules, as opposed to algorithmic. Unfortunately, algorithmic typing for full D$_{<:}$ is known to be an undecidable problem.

We explore the design space for a restricted version of D$_{<:}$ that has decidable typechecking. Even in this simplified D$_{<:}$, algorithmic typing and subtyping are tricky, due to the 
``bad bounds" problem. The Scala compiler bypasses bad bounds at the cost of a loss in  expressiveness in its type system. Based on the approach taken in the Scala compiler, we present the \emph{Step Typing} and \emph{Step Subtyping} relations for D$_{<:}$. We prove these relations sound and decidable. They are not complete with respect to the original D$_{<:}$ rules.
\end{abstract}




\maketitle

\section{Introduction}

Would you rather have a typechecker that is run by the computer but is sometimes wrong, or one that is always right but needs to be run by hand?

\emph{``I want to have my cake and eat it too"}, you say. That is going to be difficult. On the one hand, the Scala compiler implements a typechecking algorithm that accepts or rejects Scala programs, but is ocassionally wrong due to bugs. On the other hand, the DOT calculus is type-safe \cite{Amin2016}, but its typing rules can only be run manually via a proof assistant.

Why manually? The problem is that the typing rules are not syntax-directed, so an algorithm cannot be easily derived from them. For example, take the transitivity rule for subtyping, present in many calculi (DOT included):

\infrule[Trans]
 {\sub{\Gamma}{S}{T} \andalso \sub{\Gamma}{T}{U}}
 {\sub{\Gamma}{S}{U}}

For a theorem prover this rule is no problem: get the human to provide a $T$ for which the premises are satisfied, and then we can conclude $\sub{\Gamma}{S}{U}$. For an algorithm, it is harder: how should it guess the right $T$? Iterating over the infinitely many possibilities is not an option.

The standard solution is to merge the problematic rule with the other rules that use it, so that it becomes less general but more tractable. Here is how F$_{<:}$ \cite{cardelli1994extension} merges transitivity with type-variable lookup:

\infrule[Trans-TVar]
  {X <: U \in \Gamma \andalso \sub{\Gamma}{U}{T}}
  {\sub{\Gamma}{X}{T}}

This is better: to determine whether $X$ is a subtype of $T$, the typechecker can look up $X$ in $\Gamma$, obtain the upper bound $U$, and recursively check whether $\sub{\Gamma}{U}{T}$.

The algorithmic presentation of the typing rules has one potential disadvantage and one clear advantage when compared to the declarative style. The disadvantage is that for the algorithmic rules it might be less clear what programs are type-correct. Going back to the transitivity example, if we replace \rn{Trans} by \rn{Trans-TVar}, it is no longer clear whether the following program typechecks:

\begin{lstlisting}[language=Scala]
// We know that in the current environment
// Human <: Mammal and Mammal <: Organism
val orgs: List[Organism] = ...
orgs.push(new Human()) // Human <: Organism?
\end{lstlisting}

The advantage is that if the typing rules are syntax-directed,
it is possible to write an algorithm for type-checking programs. We can then implement that algorithm as, for example, a PLT Redex \cite{klein2012run} model and use the model to improve our type system.

In this paper, we describe our work in progress towards algorithmic typing for D$_{<:}$ \cite{Amin2016}, a simple calculus that is the core of DOT (Figure \ref{dsubsyntax}). Our quest gets off to a bad start: D$_{<:}$ is a generalization of F$_{<:}$, the polymorphic lambda calculus with subtyping. Typing F$_{<:}$ is undecidable \cite{pierce1994bounded}, which makes typing D$_{<:}$ also undecidable \cite{arxiv1}.

\begin{figure}
\setlength\tabcolsep{1.5pt}
\begin{tabular}{l l}
  $x, y, z$            & \textbf{Variable} \\
  $v ::=$              & \textbf{Value}    \\
  \ \ $\{A = T\}$      & \ \ type tag      \\
  \ \ $\lambda(x: T)t$ & \ \ lambda        \\
  $s, t, u ::=$        & \textbf{Term}     \\
  \ \ $x$              & \ \  variable     \\
  \ \ $v$              & \ \  value \\
  \ \ $x \ y$            & \ \  application \\
  \ \ \textbf{let} $x = t$ \textbf{in} $u$ & \ \  let \\
\\
$S, T, U ::=$ & \textbf{Type}    \\
\ \ $\top$  & \ \ top type       \\
\ \ $\bot$  & \ \ bottom type    \\
\ \ $\tTypeDec{A}{S}{T}$    & \ \ type declaration \\
\ \ $x.A$                   & \ \ path-dependent type \\
\ \ $\forall(x: S)T$        & \ \ dependent function \\
\end{tabular}
\caption{Terms and types of D$_{<:}$ \citep{Amin2016}}
\label{dsubsyntax}
\end{figure}

There is still hope, though. There are simpler versions of F$_{<:}$ with decidable typechecking. For example, Kernel F$_{<:}$, whose typing relation is decidable \cite{cardelli1985understanding}, differs only minimally from full F$_{<:}$ (Figure \ref{subffamily}). Specifically, Kernel F$_{<:}$ is less permissive when testing for subtyping of function types. Notice how in the Kernel version, the upper bound $T$ needs to be the same for both types in the conclusion; by contrast, full F$_{<:}$ allows different upper bounds. Could Kernel F$_{<:}$ be used as the basis for a simpler D$_{<:}$ that can be algorithmically typed?

\begin{figure}
\infrule[Full-S-All]
  {\sub{\Gamma}{T_2}{T_1} \andalso \sub{\Gamma, X <: T_2}{U_1}{U_2}}
  {\sub{\Gamma}{\forall(X <: T_1)U_1}{\forall(X <: T_2)U_2}}

\infrule[Kernel-S-All]
  {\sub{\Gamma, X <: T}{U_1}{U_2}}
  {\sub{\Gamma}{\forall(X <: T)U_1}{\forall(X <: T)U_2}}
   
\caption{Subtyping of function types in full F$_{<:}$ (undecidable) and Kernel F$_{<:}$ (decidable) \citep{Pierce:2002:TPL:509043}}
\label{subffamily}
\end{figure}

This paper makes three contributions:
\begin{itemize}
\item We describe how even when the original source of undecidability is eliminated, the problem of bad bounds complicates algorithmic typing and subtyping of D$_{<:}$ (Section \ref{badbounds}).

\item If bad bounds are so hard to deal with, how does the Scala compiler handle them? In fact, it does not. In Section \ref{typingscala}, we show how the Scala compiler sidesteps the bad bounds problem by using a subtyping relation that is not transitive.

\item Finally, in Section \ref{formal} we introduce the Step Typing and Subtyping relations. Step Typing and Subtyping are sound and decidable, but not complete, with respect to D$_{<:}$'s standard relations.
\end{itemize}

\section{Bad Bounds}
\label{badbounds}
\citet{Pierce:2002:TPL:509043} presents a design recipe for coming up with algorithmic typing rules for a calculus:
\begin{itemize}
\item Start with a set of declarative typing rules.
\item Modify the rules so that they are all syntax-directed.
\item Prove the syntax-directed rules sound with respect to the declarative ones. If $\vdash_A$ is the algorithmic typing relation, $\talgo{\Gamma}{t}{T} \implies \typ{\Gamma}{t}{T}$.
\item Finally, prove a \emph{minimality} result: if a term can be typed, the algorithmic rules will type it with the most precise type. $\typ{\Gamma}{t}{U} \implies \talgo{\Gamma}{t}{T} \land \sub{\Gamma}{T}{U}$.
\end{itemize}

Minimality is important. While the declarative typing rules can afford to assign arbitrarily many types to a term, the algorithmic rules need to assign just one type, for the sake of efficiency and determinism. Below, we conjecture that there does not exist an algorithmic typing relation for D$_{<:}$ that satisfies the minimality condition. 

D$_{<:}$ has a restricted form of types as values. A type tag 
{$\{A = T\}$ defines $A$ as a synonym for $T$, and has type $\tTypeDec{A}{T}{T}$. If bound in the current environment to a variable, the type tag can later be used as a path-dependent type (\code{x.A}):

\begin{lstlisting}[mathescape=true]
$\textbf{let } x = \{A = T\} \textbf{ in}$
$\ldots x.A \ldots$
\end{lstlisting}

A path-dependent type is related to the lower and upper bounds in its type declaration via subtyping:

\infrule[<:-Sel]
  {\typ{\Gamma}{x}{\tTypeDec{A}{S}{T}}}
  {\sub{\Gamma}{S}{x.A}}

\infrule[Sel-<:]
  {\typ{\Gamma}{x}{\tTypeDec{A}{S}{T}}}
  {\sub{\Gamma}{x.A}{T}}

Now notice what happens when \rn{<:-Sel} and \rn{Sel-<:} are combined with \rn{Trans}, in a term

\begin{lstlisting}[mathescape=true]
$\lambda(e: \tTypeDec{E}{\top}{\bot})$
  $\lambda(f: \top)$
    $\lambda(x: \top)$
      $f x$
\end{lstlisting}

How can $f x$ be well-typed, when $f$ has type $\top$? The reason is that the application is typed in an environment where  $\top <: e.E$ and $e.E <: \bot$, which means that $\top <: \bot$, because of \rn{Trans}. The entire type lattice collapses, so $f$ can also be assigned type e.g. $\top \rightarrow \top$, making $f x$ type-correct. 

In effect, a type declaration introduces not only a subtyping relation between a path-dependent type and its bounds, but also a subtyping relation between the bounds themselves. \citet{Amin2016} refer to these ``strange" type declarations as having \emph{bad bounds}. Bad bounds affect minimality because a term can now be typed with two different types, neither of which is a subtype of the other. This leads us to the conjecture below.

\begin{conj}[Impossibility of minimal typing]
\label{minimality}
Let $\typ{\Gamma}{t}{T}$ and $\sub{\Gamma}{T}{U}$ be the typing and subtyping relations for D$_{<:}$. There does not exist a function $\talgo{\Gamma}{t}{T}$\footnote{Notice $\talgo{\Gamma}{t}{T}$ is a function, and not simply a relation. Therefore, the $(\Gamma, t)$ pair is mapped to at most one type.} such that the following two hold:
\begin{itemize}
\item $\talgo{\Gamma}{t}{T} \implies \typ{\Gamma}{t}{T}$ (soundness)
\item $\typ{\Gamma}{t}{U} \implies \talgo{\Gamma}{t}{T} \land \sub{\Gamma}{T}{U}$ (minimality)
\end{itemize}
\end{conj}

To see why the conjecture should be true, suppose such a function $\vdash_A$ exists. Now consider the term

\begin{lstlisting}[mathescape=true]
$A \equiv \lambda(e: \tTypeDec{E}{\forall(b: B)B}{\forall(b: B)C})$
      $\textbf{let } f = \lambda(b: B)b \textbf{ in}$ 
        $\textbf{let } b = \{V = \top\} \textbf{ in }$
          $f b$
\end{lstlisting}

where $B$ and $C$ are syntactic abbrevations for types:
\begin{itemize}
\item $B \equiv \tTypeDec{V}{\top}{\top}$
\item $C \equiv \tTypeDec{Z}{\top}{\top}$
\end{itemize}

While typechecking $A$, we will eventually descend into the environment $\Gamma_{\star} = e: \tTypeDec{E}{\forall(b: B)B}{\forall(b: B)C}$. The introduced bad bounds ensure $\sub{\Gamma_{\star}}{\forall(b: B)B}{\forall(b: B)C}$. If $w$ denotes the body of the lambda, $\typ{\Gamma_{\star}}{w}{B}$ and $\typ{\Gamma_{\star}}{w}{C}$.

Minimality implies $\talgo{\Gamma_{\star}}{w}{T}$, with $\sub{\Gamma_{\star}}{T}{B}$. By Lemma \ref{declstodecls}, $T = \bot$ or $T = \tTypeDec{V}{V_1}{V_2}$. 
Similarly, $\sub{\Gamma_{\star}}{T}{C}$, which means $T = \bot$ or $T = \tTypeDec{Z}{Z_1}{Z_2}$. This means $T = \bot$. Because $\vdash_{A}$ is sound, we must have $\typ{\Gamma_{\star}}{w}{\bot}$, which does not seem like an obtainable judgement (but we are missing the proof).

\begin{rmk}
It is not clear that the impossibility result, if true, carries over to DOT, because DOT has intersection types. A typing function for DOT might be able to produce the judgement $\talgo{\gammas}{w}{B \land C}$, and $B \land C$ is plausibly a minimal typing for $w$.
\end{rmk}

In general, it is surprisingly tricky to prove statements about D$_{<:}$ that involve bad bounds. Before we tackle Lemma \ref{declstodecls}, we need to prove a rather cumbersome technical lemma.

\begin{restatable}[$\Gamma_{\star}$ is Well-Behaved]{lemma}{gammastarlemma}
\label{gammawellbehaved}
Define the following ``colour" predicates on types:

\infax{\bluep{e.E}}

\infrule
  {$T$ \text{ is a function type}}
  {\bluep{T}}
  
\infrule
  {E_1 = \bot \lor \bluep{E_1} \\
   E_2 = \top \lor \bluep{E_2}}
  {\redp{\tTypeDec{E}{E_1}{E_2}}}

The red and blue tags have no meaning beyond partitioning types into two sets that start as disjoint and stay disjoint, in the presence of subtyping.

Then all of the following hold:
\begin{align}
\typ{\gammas}{e}{T} &\implies T = \top \lor \redp{T} \label{wellbehaved1} \\
\sub{\gammas}{T}{S} \land \redp{T} &\implies S = \top \lor \redp{S} \label{wellbehaved2} \\
\sub{\gammas}{S}{T} \land \redp{T} &\implies S = \bot \lor \redp{S} \label{wellbehaved3} \\
\sub{\gammas}{T}{S} \land \bluep{T} &\implies S = \top \lor \bluep{S} \label{wellbehaved4} \\
\sub{\gammas}{S}{T} \land \bluep{T} &\implies S = \bot \lor \bluep{S} \label{wellbehaved5} \\
\sub{\gammas}{T}{\bot} &\implies T = \bot \label{wellbehaved6} \\
\sub{\gammas}{\top}{T} &\implies T = \top \label{wellbehaved7}
\end{align}
\end{restatable}

\begin{proof}
By mutual induction on a derivation of any of the statements above. The full proof is included in the appendix.
\end{proof}

We can now show that in $\Gamma_{\star}$ type declarations do not ``switch" their tags.

\begin{lemma}
\label{declstodecls}
$\sub{\Gamma_{\star}}{T}{\tTypeDec{X}{X_1}{X_2}} \implies T = \bot \lor T = \tTypeDec{X}{X_1'}{X_2'}$.
\end{lemma}

\begin{proof}
By induction on a derivation of $\sub{\Gamma_{\star}}{T}{\tTypeDec{X}{X_1}{X_2}}$, using Lemma \ref{gammawellbehaved} to reason about type bounds.
\end{proof}

\section{Typing Scala}
\label{typingscala}
If bad bounds cause so much trouble, how does the Scala compiler\footnote{By the ``Scala compiler", we mean a June 2017 version of Dotty. In most cases, \code{scalac} behaves similarly to Dotty.} manage to typecheck them? In fact, Scala avoids dealing with bad bounds by restricting its subtyping relation to not be transitive.

Consider the example code below, which is a Scala version of our 
counterexample for minimality of D$_{<:}$. In D$_{<:}$, the code would typecheck, because the bounds on the abstract type declaration mean that \code{Int $=>$ Int $<:$ Int $=>$ String}.

However, the snippet does not typecheck in Scala: there is no transitivity of subtyping!

\begin{lstlisting}[language=Scala]
trait BadBounds {
  type E >: Int => Int <: Int => String
  val f: Int => Int = (x) => x
  val f2: Int => String = f // Type Mismatch Error
  /* found: (Int => Int); required: Int => String */
} 
\end{lstlisting}

Here is another example that should typecheck, but does not:

\begin{lstlisting}[language=Scala]
trait BadBounds2 {
  type E >: Int => Int <: Int => String
  val f: Int => Int = (x) => x
  val res: String = f(42) // Type Mistmach Error
  /* found: Int; required: String */
}
\end{lstlisting}

The information about the lower and upper bounds is not entirely lost. The compiler still uses it, but only when one of the two types in the subtype check is the abstract type. This patched-up version of the code typechecks:

\begin{lstlisting}[language=Scala]
trait BadBounds {
  type E >: Int => Int <: Int => String
  val f: Int => Int = (x) => x
  (*@\hl{val e: E = f}@*)
  (*@\hl{val f2: Int $=>$ String = e}@*)
}
\end{lstlisting}

Here, the two subtype checks executed are
\begin{itemize}
\item $Int => Int <: E$?
\item $E <: Int => String$?
\end{itemize}

Both of these involve $E$ directly, and so the type bounds are considered during the check. Indeed, inspection of the Dotty code shows it runs an algorithm similar to the one in Figure \ref{dottysub}.

\begin{figure}[h]
\begin{lstlisting}[language=Scala]
def sub(t1: Type, t2: Type): Boolean = {
  val f = t2 match {
    case TypeBounds(l2, u2) => sub(t1, l2)
    ...
  }
  f || t1 match {
    case TypeBounds(l1, u1) => sub(u1, t2)
    ...
  }
}
\end{lstlisting}
\caption{Dotty's subtyping algorithm (baby version)}
\label{dottysub}
\end{figure}

In addition to dropping transitivity, Scala's handling of bad bounds takes exponential time in the worst case. Let $P_N$ denote the following program:

\begin{lstlisting}[language=Scala, escapeinside={/*}{*/}]
trait Foo {
  type /*$T_1 <: T_2$*/; /*$\ldots$*/; type /*$T_{N - 1} <: T_N$*/; type /*$T_N$*/
  type /*$T_{2*N} >: T_{2*N - 1}$*/; /*$\ldots$*/; type /*$T_{N + 2} >: T_{N + 1}$*/; type /*$T_{N + 1}$*/
  val /*$v_1: T_1$*/; val /*$v_{2*N}: T_{2*N} = v_1$*/
}
\end{lstlisting} 

Notice that $T_1 <: T_N$ via a chain of $N$ upper bounds that are discoverable by the subtyping algorithm. The same holds for $T_{N + 1} <: T_{2 * N}$ via lower bounds. However, $T_N$ is not a subtype of $T_{N + 1}$, so a subtype check \code{sub$(T_1, T_{2*N})$} will fail only after at least $N$ nested recursive calls. Since all the calls are eventually unsuccessful, this means there are at least $2^N$ recursive calls. If we plot how long it takes to compile $P_N$ for different values of $N$, we can see that the time increases exponentially (Figure \ref{subexpo}).

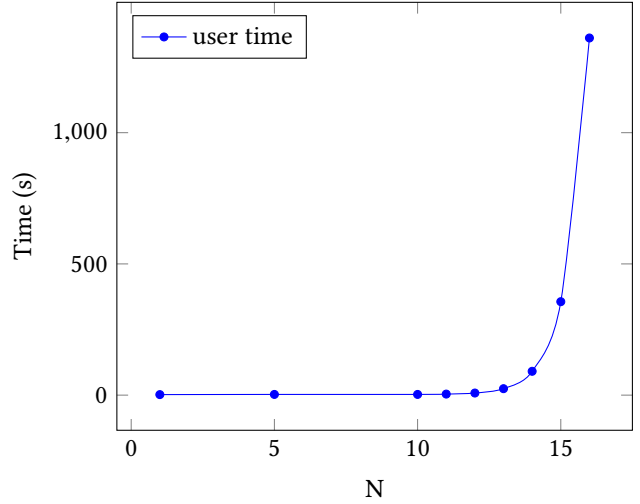
\begin{figure}
\begin{tikzpicture}
    \begin{axis}[
        xlabel=N,
        ylabel=Time (s),
        legend style={
          cells={anchor=east},
          legend pos=north west,
        },
        mark size=1.5pt,
    ]
    \addplot[smooth,mark=*,blue] plot coordinates {
        (1, 2)
        (5, 2.7)
        (10, 2.7)
        (11, 4.0)
        (12, 8.0)
        (13, 24.6)
        (14, 90.8)
        (15, 356.1)
        (16, 1360.8)
    };
    \legend{user time}
    \end{axis}
\end{tikzpicture}
\caption{Compiling $P_N$ takes exponential time}
\label{subexpo}
\end{figure}

\section{Formalization}
\label{formal}

In this section, we formalize our approach to algorithmic typing and subtyping for D$_{<:}$. We first define three helper relations: Exposure, Promotion, and Demotion. Using these relations, we then present Step Typing and Step Subtyping, which form a sound, decidable typechecking algorithm for a subset of D$_{<:}$. 

\subsection{Preliminaries}
\label{preliminaries}

We start by making two simplifying assumptions, without loss of generality. First, we use Barendregt's Variable Convention, to avoid having to manually specify $\alpha$-conversions. Second, all our type environments $\Gamma$ are assumed to be \emph{well-formed}, which is implied by the presentation of D$_{<:}$ in \citet{Amin2016} (D$_{<:}$ does not have recursive types). The rules for judgements of well-formedness ($\wfenv{\Gamma}$) are shown in Figure \ref{wellformed}.

\begin{figure}
\begin{flushright}
\fbox{$\wfenv{\Gamma}$}
\end{flushright}

\infrule[W-Cons]
  {\wfenv{\Gamma} \andalso x \not \in \dom{\Gamma} \andalso x \not \in fv(T)}
  {\wfenv{\Gamma, x: T}}
  
\infax[W-Empty]
  {\wfenv{\varnothing}}
\caption{Well-formed environments}
\label{wellformed}
\end{figure}

\subsection{Exposure}
\label{exposure}

The \emph{Exposure }relation $\expos{\Gamma}{T}{T'}$ (Figure \ref{exposrel}) ``gets rid" of path-dependent types. It will later be used in places where the typechecker sees a path-dependent type, but needs a supertype of it that is a function or a type declaration.

We base our Exposure relation in both the Exposure operation present in Kernel F$_{<:}$ \citep{Pierce:2002:TPL:509043} and the treatment of type bounds in Scala. In \citet{Pierce:2002:TPL:509043}, Exposure gives us the least supertype that is not a type variable. We conjecture that the result does not carry through to D$_{<:}$, because in D$_{<:}$ a path-dependent type can be a subtype of multiple types that are unrelated by subtyping:

\begin{lstlisting}[mathescape=true]
$\lambda(a: \tTypeDec{A}{\bot}{\top})$
  $\lambda(b: \tTypeDec{B}{a.A}{\tTypeDec{B_1}{\bot}{\top}})$
    $\lambda(c: \tTypeDec{C}{a.A}{\tTypeDec{C_1}{\bot}{\top}})$
      $w$
\end{lstlisting}

Within $w$, $a.A <: \tTypeDec{B_1}{\bot}{\top}$ and $a.A <: \tTypeDec{C_1}{\bot}{\top}$, but there is probably no subtyping between $\tTypeDec{B_1}{\bot}{\top}$ and $\tTypeDec{C_1}{\bot}{\top}$.

Exposure preserves subtyping (Lemma \ref{exposound}), and terminates (Lemma \ref{exposterm}).

\begin{lemma}[Exposure preserves subtyping]
\label{exposound}
$\expos{\Gamma}{T}{T'} \implies \sub{\Gamma}{T}{T'}$
\end{lemma}

\begin{proof}
By induction on a derivation of $\expos{\Gamma}{T}{T'}$.

%
%
%
%
%
\end{proof}

\begin{lemma}[Termination of Exposure]
\label{exposterm}
When viewed as an algorithm, exposure terminates.
\end{lemma}
\begin{proof}
Follows from the fact that  $\Gamma = \Gamma_1, x: T, \Gamma_2$ means that $x \not \in fv(T)$ and $x$ does not show up free in $\Gamma_1$, since $\wfenv{\Gamma}$ (well-formed environments do not have cycles).
\end{proof}

\begin{figure}
\begin{flushright}
\fbox{$\expos{\Gamma}{T}{T'}$}
\end{flushright}

\infrule[X-Bot]
  {\Gamma(x) = T \andalso \expos{\Gamma}{T}{\bot}}
  {\expos{\Gamma}{x.A}{\bot}}
\infrule[X-Path]
  {\Gamma(x) = T \andalso \expos{\Gamma}{T}{\tTypeDec{A}{S}{U}} \andalso \expos{\Gamma}{U}{V}}
  {\expos{\Gamma}{x.A}{V}}

\infrule[X-Other]
  {T \mbox{ is not a path-dependent type}}
  {\expos{\Gamma}{T}{T}}
  
\caption{Exposure}
\label{exposrel}
\end{figure}

\subsection{Promotion and Demotion}
\label{promodemo}

In Section \ref{steptyping}, we will sometimes need to remove all references to a specific variable from a type. The \emph{Promotion} (Figure \ref{promoop}) and \emph{Demotion} relations (Figure \ref{demoop}), adapted from \citet{pierce2000local}, accomplish this. They remove all occurrences of the specified free variable from a type (Lemma \ref{correctpromodemo}), preserve subtyping (Lemma \ref{soundpromodemo}), and terminate (Lemma \ref{termpromodemo}).

\begin{figure}
\begin{flushright}
\fbox{$\promo{\Gamma}{T}{x}{T'}$}
\end{flushright}

\infrule[P-Up]
  {\Gamma(x) = T \andalso \expos{\Gamma}{T}{\tTypeDec{A}{L}{U}}}
  {\promo{\Gamma}{x.A}{x}{U}}

\infrule[P-Up-Bot]
  {\Gamma(x) = T \andalso \expos{\Gamma}{T}{\bot}}
  {\promo{\Gamma}{x.A}{x}{\bot}}

\infrule[P-Lam]
  {\demo{\Gamma}{S}{x}{S'} \andalso \promo{\Gamma, y: S'}{T}{x}{T'} \andalso y \ne x}
  {\promo{\Gamma}{\forall(y: S)T}{x}{\forall(y: S')T'}}  
  
\infrule[P-Var]
  {y \ne x}
  {\promo{\Gamma}{y.A}{x}{y.A}}
  
\infax[P-Bot]
  {\promo{\Gamma}{\bot}{x}{\bot}}

\infax[P-Top]
  {\promo{\Gamma}{\top}{x}{\top}} 
  
\infrule[P-Decl]
  {\demo{\Gamma}{L}{x}{L'} \andalso \promo{\Gamma}{U}{x}{U'}}
  {\promo{\Gamma}{\tTypeDec{A}{L}{U}}{x}{\tTypeDec{A}{L'}{U'}}} 

\infax[P-Cap]
  {\promo{\Gamma}{\forall(x: S)T}{x}{\forall(x: S)T}}

\caption{Promotion}
\label{promoop}
\end{figure}

\begin{figure}
\begin{flushright}
\fbox{$\demo{\Gamma}{T}{x}{T'}$}
\end{flushright}

\infrule[D-Down]
  {\Gamma(x) = T \andalso \expos{\Gamma}{T}{\tTypeDec{A}{L}{U}}}
  {\demo{\Gamma}{x.A}{x}{L}}

\infrule[D-Down-Bot]
  {\Gamma(x) = T \andalso \expos{\Gamma}{T}{\bot}}
  {\demo{\Gamma}{x.A}{x}{\top}} 

\infrule[D-Lam]
  {\promo{\Gamma}{S}{x}{S'} \andalso \demo{\Gamma, y: S}{T}{x}{T'} \\
  y \ne x}
  {\demo{\Gamma}{\forall(y: S)T}{x}{\forall(y: S')T'}}
  
\infrule[D-Var]
  {y \ne x}
  {\demo{\Gamma}{y.A}{x}{y.A}}
  
\infax[D-Bot]
  {\demo{\Gamma}{\bot}{x}{\bot}}

\infax[D-Top]
  {\demo{\Gamma}{\top}{x}{\top}}
  
\infrule[D-Decl]
  {\promo{\Gamma}{L}{x}{L'} \andalso \demo{\Gamma}{U}{x}{U'}}
  {\demo{\Gamma}{\tTypeDec{A}{L}{U}}{x}{\tTypeDec{A}{L'}{U'}}}

\infax[D-Cap]
  {\demo{\Gamma}{\forall(x: S)T}{x}{\forall(x: S)T}}
  
\caption{Demotion}
\label{demoop}
\end{figure}

\begin{lemma}[Correctness of Promotion and Demotion]
\label{correctpromodemo}
If $\promo{\Gamma}{T}{x}{T'}$ or $\demo{\Gamma}{T}{x}{T'}$, then $x \not \in fv(T')$.
\end{lemma}

\begin{proof}
By induction on a derivation of $\promo{\Gamma}{T}{x}{T'}$ or $\demo{\Gamma}{T}{x}{T'}$.
\end{proof}

\begin{lemma}[Promotion and Demotion preserve subtyping]
\label{soundpromodemo}
$\promo{\Gamma}{T}{x}{T'} \implies \sub{\Gamma}{T}{T'}$ and $\demo{\Gamma}{T}{x}{T'} \implies \sub{\Gamma}{T'}{T}$
\end{lemma}

\begin{proof}
By induction on a derivation of $\promo{\Gamma}{T}{x}{T'}$ or $\demo{\Gamma}{T}{x}{T'}$.
\end{proof}

\begin{lemma}[Termination of Promotion and Demotion]
\label{termpromodemo}
When viewed as an algorithm, both promotion and demotion terminate.
\end{lemma}

\begin{proof}
Uses Lemma \ref{exposterm}. The size of the term we are promoting or demoting is a termination measure.
\end{proof}

\subsection{Step Typing}
\label{steptyping}
We can now define Step Typing. The typing rules are shown in Figure \ref{steptypingrules}. Differences with the calculus in \citet{Amin2016} are highlighted.

There are two rules in the standard typing relation that are not syntax-directed: \rn{Sub}, which is needed when typing function applications, and \rn{Let}, for typing let-expressions.

\infrule[Sub]
  {\typ{\Gamma}{t}{T} \andalso \sub{\Gamma}{T}{U}}
  {\typ{\Gamma}{t}{U}}

\infrule[Let]
  {\typ{\Gamma}{t}{T} \andalso \typ{\Gamma, x: T}{u}{U} \\
  x \not \in fv(U)}
  {\typ{\Gamma}{\textbf{let } x = t \textbf{ in } u: U}}

\text{\rn{Sub}} is not syntax-directed because the typechecker needs to ``guess" the type $U$ in the conclusion. Similarly, \rn{Let} forces us to guess a type $U$ where $x$ is not free.

To fix these issues, Step Typing differs from the standard typing relation in two ways:
\begin{itemize}
\item It drops the subsumption rule: instead, when typing a function application, Step Typing uses Exposure to find a function type (or $\bot$) for the term in the function position.
\item Additionally, it uses Promotion to remove all references to the bound variable in the returned type of a let-expression.
\end{itemize}

\begin{figure}
\begin{flushright}
\fbox{$\styp{\Gamma}{t}{T}$}
\end{flushright}

\infrule[T-Var]{\Gamma(x) = T}{\styp{\Gamma}{x}{T}}

\infrule[T-All-I]
  {\styp{\Gamma, x \colon T}{t}{U} \andalso x \not \in \text{fv}(T)}
  {\styp{\Gamma}{\lambda (x \colon T) t}{\forall (x \colon T) U}}

\infax[T-Typ-I]
  {\styp{\Gamma}{\{A = T\}}{\tTypeDec{A}{T}{T}}}

\newruletrue
\infrule[T-All-E]
  {\styp{\Gamma}{x}{V} \andalso \expos{\Gamma}{V}{\forall (z \colon S)T} \\ \styp{\Gamma}{y}{U} \andalso \ssub{\Gamma}{U}{S}}
  {\styp{\Gamma}{x \ y}{[z := y] T}}
\newrulefalse

\newruletrue
\infrule[T-App-Bot]
  {\styp{\Gamma}{x}{T} \andalso \expos{\Gamma}{T}{\bot} \\ \styp{\Gamma}{y}{U}}
  {\styp{\Gamma}{x \ y}{\bot}}
\newrulefalse

\newruletrue
\infrule[T-Let]
  {\styp{\Gamma}{t}{T} \andalso \styp{\Gamma, x \colon T}{u}{U'} \\
  \promo{\Gamma, x \colon T}{U'}{x}{U}}
  {\styp{\Gamma}{\textbf{let} \  x = t \ \textbf{in} \ u}{U}}
\newrulefalse

\caption{Step Typing}
 
\label{steptypingrules}
\end{figure}

\subsection{Step Subtyping}
\label{stepsubtyping}
The standard subtyping relation requires three changes: the first two to make the rules syntax-directed, and the last one to guarantee termination:
\begin{itemize}
\item We drop the general reflexivity rule, replacing it with reflexivity of only path-dependent types. General reflexivity still holds, just not as an axiom (Lemma \ref{reflexsub}).

\item Transitivity goes away: instead, we use Exposure when comparing path-dependent types (like in Scala).

\item So that the algorithm terminates, we only allow subtyping between function types with the same argument type (as opposed to the standard contravariant rule). This is the same restriction used to make Kernel F$_{<:}$ decidable \citet{cardelli1985understanding}.
\end{itemize}

The rules for Step Subtyping are shown in Figure \ref{stepsubrules}.

\begin{figure}
\begin{flushright}
\fbox{$\ssub{\Gamma}{S}{T}$}
\end{flushright}

\infax[S-Bot]{\ssub{\Gamma}{\bot}{T}}

\infax[S-Top]{\ssub{\Gamma}{T}{\top}}

\infrule[S-Typ-$<:$-Typ]
  {\ssub{\Gamma}{S_2}{S_1} \andalso \ssub{\Gamma}{T_1}{T_2}}
  {\ssub{\Gamma}{\tTypeDec{A}{S_1}{T_1}}{\tTypeDec{A}{S_2}{T_2}}}

\newruletrue
\infax[S-Refl]{\ssub{\Gamma}{x.A}{x.A}}
\newrulefalse

\newruletrue
\infrule[S-$<:$-Sel]
  {\Gamma(x) = T \andalso \expos{\Gamma}{T}{\tTypeDec{A}{S_1}{S_2}} \andalso \ssub{\Gamma}{S_2}{U}}
  {\ssub{\Gamma}{x.A}{U}}  
\newrulefalse

\newruletrue  
\infrule[S-Sel-$<:$]
  {\Gamma(x) = T \andalso \expos{\Gamma}{T}{\tTypeDec{A}{S_1}{S_2}} \andalso \ssub{\Gamma}{U}{S_1}}
  {\ssub{\Gamma}{U}{x.A}}  
\newrulefalse

\newruletrue  
\infrule[S-Bot-$<:$]
  {\Gamma(x) = T \andalso \expos{\Gamma}{T}{\bot}}
  {\ssub{\Gamma}{U}{x.A}}  
\newrulefalse

\newruletrue
\infrule[S-$<:$-Bot]
  {\Gamma(x) = T \andalso \expos{\Gamma}{T}{\bot}}
  {\ssub{\Gamma}{x.A}{U}}
\newrulefalse
  
\newruletrue
\infrule[S-All-$<:$-All]
  {\ssub{\Gamma, x \colon S}{T_1}{T_2}}
  {\ssub{\Gamma}{\forall (x \colon S) T_1}{\forall (x \colon S) T_2}}
\newrulefalse

\caption{Step Subtyping}
\label{stepsubrules}
\end{figure}

\subsection{Metatheoretic Properties}
\label{metaprops}

We now summarize the metatheoretic properties of Step Typing and Subtyping:
\begin{itemize}

\item \textbf{Soundness}: Step Typing and Subtyping are sound with respect to the standard typing and subtyping relations of D$_{<:}$ (Theorem \ref{thmsoundness}).

\item \textbf{Decidability}: Both relations are decidable (Theorem \ref{typandsubdecidable}).

\item \textbf{Completeness}: the relations are not complete. In fact, no algorithm relation can be complete, since typing D$_{<:}$ is undecidable. Any program that relies on a combination of bad bounds and transitivity to typecheck will fail to do so.

\item \textbf{Subject Reduction}: we do not currently know whether the subject-reduction property holds for Step Typing. This means we could have $\styp{\Gamma}{t}{T}$ and $t \red t'$, but $t'$ can only be typed under the standard typing relation, and not Step Typing.
\end{itemize}

These results are formalized below.

The weight function, adapted from \citet{Pierce:2002:TPL:509043}}, will serve as a termination measure for subtyping.

\begin{defn}[Weight]
The weight of type $T$ in context $\Gamma$, written $\weight{\Gamma}{T}$, is given by the equations below:
\[
\begin{array}{l}
  \weight{\Gamma}{\top} = 1 \\
  \weight{\Gamma}{\bot} = 1 \\
  \weight{\Gamma}{\tTypeDec{A}{S}{T}} = 1 + \max(\weight{\Gamma}{S}, \weight{\Gamma}{T}) \\
  \weight{\Gamma}{x.A} = 1 + \weight{\Gamma_1}{T} \text{ if } \Gamma = \Gamma_1, x: T, \Gamma_2 \\
  \weight{\Gamma}{\forall(x: S)T} = 1 + \weight{\Gamma, x: S}{T}
\end{array}
\]
\end{defn}

\begin{lemma}[Monotonicity of Exposure]
\label{expo-weight-mono}
$\expos{\Gamma}{T}{T'} \implies \weight{\Gamma}{T} \ge \weight{\Gamma}{T'}$
\end{lemma}

\begin{proof}
By induction on a derivation of $\expos{\Gamma}{T}{T'}$, using $\wfenv{\Gamma}$.
\end{proof}

\begin{theorem}[Soundness of Step Typing and Subtyping]{thm}
\label{thmsoundness}
$\styp{\Gamma}{t}{T} \implies \typ{\Gamma}{t}{T}$ and $\ssub{\Gamma}{S}{U} \implies \sub{\Gamma}{S}{U}$
\end{theorem}

\begin{proof}
By mutual induction on a derivation of $\styp{\Gamma}{t}{T}$ or $\ssub{\Gamma}{S}{U}$. We only need to consider the new rules.

%
%
%
%
%
%
%
\end{proof}

\begin{theorem}[Decidability of Step Typing and Subtyping]
\label{typandsubdecidable}
 Step Typing and Subtyping are decidable.
\end{theorem}

\begin{proof}
Since Step Typing and Subtyping are syntax-directed, to prove decidability we need to argue that they terminate.

The size of the term under consideration is a termination measure for Step Typing. For Step Subtyping, define $\weightsub{\Gamma}{S}{T} = \weight{\Gamma}{S} + \weight{\Gamma}{T}$. We can show that in the Step Subtyping rules, the weight of the conclusions is always strictly greater than the weight of the premises.

%
%
%
%
%
\end{proof}

\begin{lemma}[Reflexivity of Step Subtyping]
\label{reflexsub}
$\forall T, \ssub{\Gamma}{T}{T}$
\end{lemma}

\begin{proof}
By induction on a derivation of $\ssub{\Gamma}{T}{T}$.
\end{proof}

\section{Related Work}
\textbf{F$_{<:}$}:  \citet{pierce1994bounded} showed that algorithmic subtyping for F$_{<:}$ is undecidable. Kernel F$_{<:}$ \citep{cardelli1985understanding} introduced the Exposure operation \citep{Pierce:2002:TPL:509043}, and \citet{pierce2000local} uses the Promotion and Demotion operations to do local type inference on Kernel F$_{<:}$.

\textbf{D$_{<:}$}: \citet{arxiv1} introduced and proved D$_{<:}$ sound. The version of D$_{<:}$ we use  comes from \citet{Amin2016}, and uses ANF and small-step semantics.

\textbf{DOT}: on top of D$_{<:}$, DOT adds features like recursive and intersection types. There are many presentations of DOT, some of them differing in which features are included in the calculus, whether the operational semantics are small-step or big-step, whether ANF is used, etc. All of these use a declarative (as opposed to algorithmic) presentation of the DOT type system: \citep{amin2012dependent, oopsla14, arxiv1, Amin2016, rapoport2017mutable, amin2017type, rapoport2017simple}.

\textbf{Featherweight Scala}: \citet{fs} introduced Featherweight Scala (FS$_{\text{alg}}$), which formalizes a subset of the Scala type system. They show that the calculus has decidable typing and subtyping. FS$_{\text{alg}}$ has not been proven type-safe. Featherweight Scala is neither a subset nor a superset of D$_{<:}$, and differs from D$_{<:}$ in multiple ways: it is a class-based calculus with nominal typing and has call-by-name semantics. More relevant to our work, type members in FS$_{\text{alg}}$ (which correspond to type declarations and type tags in D$_{<:}$) are either completely abstract (\code{type A}) or aliases (\code{type A = T}). It is not possible to assign lower or upper bounds to an abstract type member (\code{type A >: S <: T}), which is possible both in Scala and D$_{<:}$. Because bounds cannot be specified, it is not possible to create a custom subtyping lattice in FS$_{\text{alg}}$, so there is no bad bounds problem.

\textbf{Scala}: the Scala type system has been shown to be both unsound \citep{null} and undecidable \citep{bjarnason2009lc, bjarnason2011ski}. Because Scala's type system is not formally specified, it is hard to say at any one point in time whether a specific proof of undecidability (or unsoundness) is still valid or not \citep{sdts}.

\section{Conclusions}

This paper described our work in progress towards a version of D$_{<:}$ with algorithmic typing. We showed how a combination of bad bounds and transitivity make it unlikely that a typing algorithm satisfying the minimality condition exists for D$_{<:}$, even after removing the known source of undecidability. We also showed how the Scala compiler deals with bad bounds by dropping transitivity of subtyping. Finally, we used prior work on decidable versions of F$_{<:}$, as well as the approach taken in the Scala compiler, to develop Step Typing and Subtyping. These relations are sound and decidable, but not complete, with respect to the standard relations.

Is the subset of D$_{<:}$ that Step Typing can type interesting? Maybe. We think a more conclusive answer will depend on whether the subject reduction property holds for Step Typing. Because Step Typing mimics the behaviour of the Scala compiler, we conjecture that the lack of transitivity does not, on its own, mean we cannot type ``useful" programs (every single Scala program written to-date has been typed with a similar restriction in place).

Future work will involve establishing subject reduction, and extending Step Typing to DOT. In doing so, we will face additional challenges because of DOT's increased complexity: recursive types and type environments that are not well-formed (they can have cycles) are among DOT's features that will be problematic.

\begin{acks}
We would like to thank Marianna Rapoport, Ifaz Kabir, Paul He, and Prabhakar Ragde for proofreading this paper, as well as for helpful comments and discussions about DOT.
\end{acks}

\bibliography{bibliography}

\appendix
\section{Proof of Lemma \ref{gammawellbehaved}}

\gammastarlemma*

\begin{proof}
By mutual induction on a derivation of any of the statements above.

\textbf{(\ref{wellbehaved1})} We can only type $e$ through \rn{Var} or \rn{Sub}. In the \rn{Var} case, certainly $\redp{\Gamma_{\star}(e)}$. If we are in the \rn{Sub} case, then we can use the induction hypothesis and (\ref{wellbehaved2}).

\textbf{(\ref{wellbehaved2})}

Case \rn{Top}: we are done, since $S = \top$.

Case \rn{Bot}: does not apply, because $T$ is a type declaration.

Case \rn{Refl}: trivial.

Case \rn{Trans}: $\sub{\Gamma_{\star}}{T}{U}$ and $\sub{\Gamma_{\star}}{U}{S}$. By the induction hypothesis, $U = \top$ or $\redp{U}$. If $U = \top$, then $T = \top$ by (\ref{wellbehaved7}). If $\redp{U}$, then we can apply the induction hypothesis one more time to get what we want about $T$.

Case \rn{<:-Sel}: $\sub{\Gamma_{\star}}{T}{x.A}$  and $\typ{\Gamma_{\star}}{x}{\tTypeDec{A}{T}{T'}}$. We must have $x = e$ and $A = E$, and so by (\ref{wellbehaved1}) $\redp{\tTypeDec{E}{T}{T'}}$. Then $T = \bot$ or $\bluep{T}$. Both lead to a contradiction because $T$ is a type declaration.

Case \rn{Sel-<:}: does not apply, because $T$ is a type declaration.

Case \rn{All-<:-All}: ditto, does not apply.

Case \rn{Typ-<:-Typ}: $T = \tTypeDec{A}{T_1}{T_2}$ and $S = \tTypeDec{A}{S_1}{S_2}$. Since $\redp{T}$, $T_1 = \bot$ or $\bluep{T_1}$. If $T_1 = \bot$, then $\sub{\Gamma_{\star}}{S_1}{T_1} \implies S_1 = \bot$, by (\ref{wellbehaved6}). If $\bluep{T_1}$, then (\ref{wellbehaved5}) means that $S_1 = \bot$ or $\bluep{S_1}$. Similarly, $T_2 = \top$ or $\bluep{T_2}$. If $T_2 = \top$, then (\ref{wellbehaved7}) means that $S_2 = \top$. If $\bluep{T_2}$, then (\ref{wellbehaved4}) implies $S_2 = \top$ or $\bluep{S_2}$. In any case, we end up with $\redp{\tTypeDec{A}{S_1}{S_2}}$, as needed.

\textbf{(\ref{wellbehaved3})} Similar to (\ref{wellbehaved2}) 

\textbf{(\ref{wellbehaved4})}

Case \rn{Top}: trivial.

Case \rn{Bot}: does not apply: $\bluep{T}$, so $T$ can't be $\bot$.

Case \rn{Refl}: trivial.

Case \rn{Trans}: $\sub{\Gamma_{\star}}{T}{U}$ and $\sub{\Gamma_{\star}}{U}{S}$. By the induction hypothesis, $U = \top$ or $\bluep{U}$. If $U = \top$, then (\ref{wellbehaved7}) implies $S = \top$. If $\bluep{U}$, then we can apply the induction hypothesis.

Case \rn{<:-Sel}: By (\ref{wellbehaved1}), we must have $\sub{\Gamma_{\star}}{T}{e.E}$, and by definition $\bluep{e.E}$.

Case \rn{Sel-<:}: We have $\sub{\Gamma_{\star}}{e.E}{S}$ and $\typ{\Gamma_{\star}}{e}{\tTypeDec{A}{S'}{S}}$. By (\ref{wellbehaved1}) and the definition of \rn{red}, we get that $\bluep{S}$.

Case \rn{All-<:-All}: trivial because all function types are blue.

Case \rn{Typ-<:-Typ}: does not apply, because a type declaration is not blue.

\textbf{(\ref{wellbehaved5})} Similar to (\ref{wellbehaved4}).

\textbf{(\ref{wellbehaved6})} The only interesting cases are \rn{Trans} and \rn{Sel-<:}. \rn{Trans} follows from two applications of the induction hypothesis. In \rn{Sel-<:}, (\ref{wellbehaved1}) implies that $\typ{\Gamma_{\star}}{e.E}{\bot}$, with $\typ{\Gamma_{\star}}{e}{\tTypeDec{E}{E_1}{\bot}}$. We know that $\redp{\tTypeDec{E}{E_1}{\bot}}$ from (\ref{wellbehaved1}), but this is a contradiction because red types can only have $\top$ or a blue type as an upper bound.

\textbf{(\ref{wellbehaved7})} Similar to (\ref{wellbehaved6}).
\end{proof}
\end{document}